\documentclass{revtex4}
\usepackage{amsmath,amsthm}
\usepackage{rotating}
\usepackage{multirow}
\usepackage{tikz}
\usetikzlibrary{arrows,shapes.gates.logic.US,shapes.gates.logic.IEC,calc}
\usetikzlibrary{circuits.logic.US}
\usepackage{circuitikz}

\usepackage{graphicx}

\usepackage{amsfonts}
\newtheorem{theorem}{Theorem}[section]

\newtheorem{proposition}[theorem]{Proposition}

\theoremstyle{remark}
\newtheorem{remark}[theorem]{Remark}

\theoremstyle{definition}

\theoremstyle{example}
\newtheorem{example}[theorem]{Example}

\theoremstyle{notation}

\newcommand{\bra}[1]{\langle#1|}
\newcommand{\ket}[1]{|#1\rangle}

\begin{document}

\title{Selective correlations in finite quantum systems and the Desargues property}            
\author{C. Lei, A. Vourdas\\
Department of Computer Science\\University of Bradford\\ Bradford BD7 1DP, UK}

\begin{abstract}
The Desargues property is well known in the context of projective geometry.
An analogous property is presented in the context of both classical and Quantum Physics.
In a classical context, the  Desargues property
implies that two logical circuits with the same input, show in their outputs selective correlations.
In general their outputs are uncorrelated, but if the output of one has a particular value,
then the output of the other has another particular value. 
In a quantum context, the  Desargues property implies that two experiments each of which involves two successive projective measurements, have selective correlations. For a particular set of projectors,
if in one experiment the second measurement does not change the output of the first measurement, then the same is true in the other experiment.
\end{abstract}

\maketitle

\section{Introduction}

Quantum logic has been studied extensively in the literature,
after the fundamental work of Birkhoff and von Neumann\cite{LO1,LO2,LO3,LO4,LO5,LO6}.
It is described with the Birkhoff-von Neumann lattice
of the closed subspaces
of the Hilbert space, with the operations of conjunction, disjunction and complementation.
In the case of the infinite-dimensional Hilbert space $H_{\rm osc}$ of the harmonic oscillator, the Birkhoff-von Neumann lattice ${\cal L}(H_{\rm osc})$ is orthomodular.

We consider quantum systems with a $d$-dimensional Hilbert space $H(d)$\cite{vour,vou}. In this case all subspaces are closed, and the lattice of subspaces ${\cal L}(d)$ is 
a modular orthocomplemented lattice \cite{la1,la2,la3,la4,la5}.
The modular orthocomplemented lattice  ${\cal L}(d)$ obeys modularity (which is a weak version of distributivity), but the orthomodular lattice ${\cal L}(H_{\rm osc})$ violates modularity. Orthomodularity is a weaker concept than orthocomplemented modularity.
Both lattices ${\cal L}(H_{\rm osc})$ and ${\cal L}(d)$, violate distributivity.

It has been pointed out in \cite{A0,A1,A2,A3,A4} and in \cite{la3}, that the lattices describing finite quantum systems, are a special case of modular orthocomplemented lattices, because they have extra stronger properties.
One such property is the Desargues property which is fundamental in Projective Geometry\cite{pro}.
The analogue of this geometrical property is discussed here, in the context of Physics, as follows:
\begin{itemize}
\item
In a classical context, the  Desargues property is expressed in the language of Boolean algebra (proposition \ref{pro1}).
Based on this we present two logical circuits with {\em selective correlations}, in the sense when the output of one has a particular value,
then the output of the other has another particular value. But in general the two outputs are not correlated. 
\item
In a quantum context, the  Desargues property is expressed in terms of subspaces of $H(d)$ and the corresponding projectors (proposition \ref{pro2}).
Based on this we present two experiments, each of which involves two successive projective measurements.
If in one experiment the second measurement does not change the output of the first measurement, then the same is true in the other experiment.
This is true only for a particular set of projectors, defined by the Desargues property.
In this sense the two experiments show {\em selective correlations}. 
\end{itemize}

\section{The Desargues property in Boolean algebra and selective classical correlations}
Classical Physics is described with Boolean algebra\cite{B1,B2}.
In the powerset $2^S$ of a finite set $S$ (i.e., on the set of the subsets of $S$), we define 
the disjunction (logical OR), conjunction (logical AND) and complementation (logical NOT) operations as the union, intersection and complement of set theory, correspondingly:
\begin{eqnarray}
A\vee B=A \cup B;\;\;\;A\wedge B=A\cap B;\;\;\;\neg A=S\setminus A;\;\;\;A,B\subseteq S.
\end{eqnarray}
The powerset $2^S$ with these operations is a Boolean algebra. 
The least element is the empty set $0=\emptyset$, and the greatest element is ${I}=S$.
The partial order $\prec$ in this lattice is the subset $\subseteq$.

Most of the work with Boolean algebras considers the binary case, where $S=\{1\}$, and $0=\emptyset$ and $I=\{1\}$.
Work with bigger sets $S$ has been considered in \cite{S1} in a classical context, and \cite{S2,S3} in a quantum context.
Here $S$ is a general finite set.

\begin{figure}[p]
    \centering
    \includegraphics[width=0.6\textwidth]{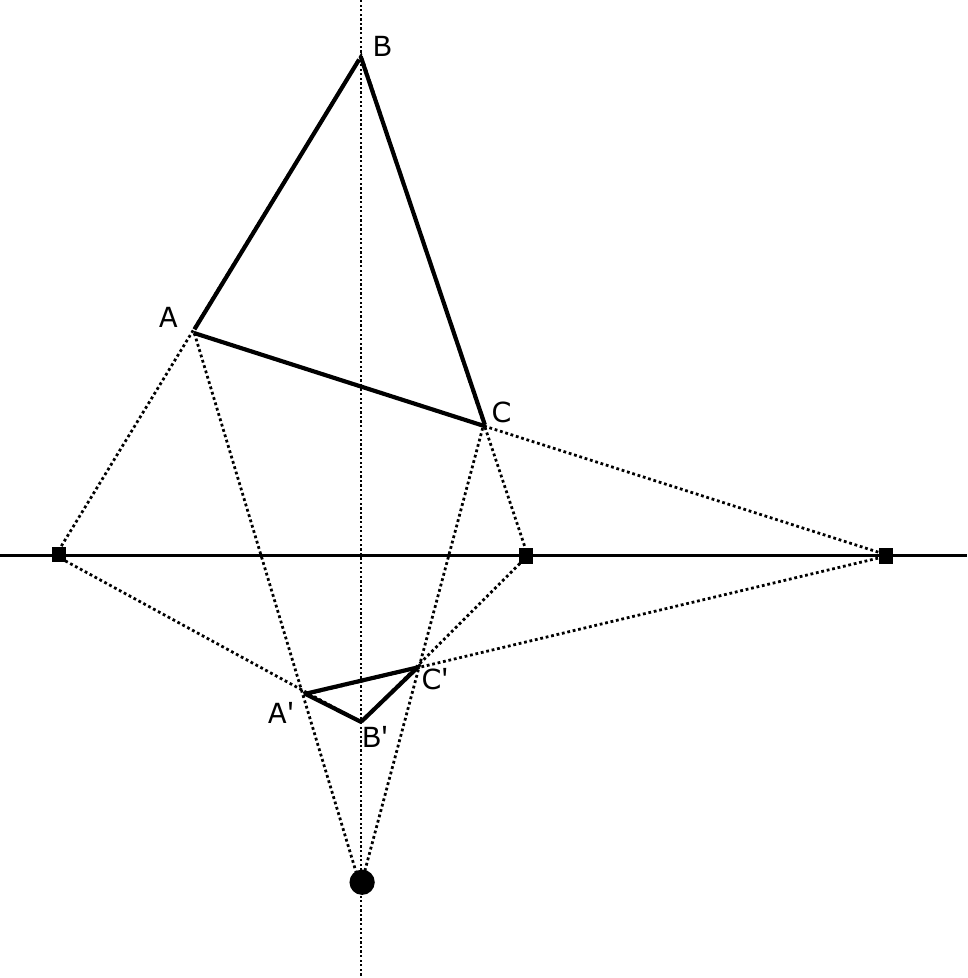}
    \caption{The Desargues property in projective geometry}
    \label{fig1}
\end{figure}

The Desargues property is fundamental in projective geometry\cite{pro}, and is shown in fig.\ref{fig1}.
Two triangles $ABC$ and $A^\prime B^\prime C^\prime$ are considered, such that the lines 
$AA^\prime$, $BB^\prime$, $CC^\prime$, intersect at the same point.
Then the three points which are the intersections of the lines $(AB, A^\prime B^\prime)$,
$(AC, A^\prime C^\prime)$, $(BC, B^\prime C^\prime)$, belong to the same line.
The analogue of this in Boolean algebra, is given in the following proposition.
\begin{proposition}\label{pro1}
Let $(A_1, A_2, A_3)$ and $(A_1^\prime, A_2^\prime , A_3^\prime)$ be two triplets of distinct subsets of $S$.
Also let 
\begin{eqnarray}
B_{ij}=A_i\vee A_j;\;\;\;B_{ij}^\prime=A_i^\prime\vee A_j^\prime;\;\;\;i\ne j
\end{eqnarray}
and 
\begin{eqnarray}
{\mathfrak B}_k=B_{ij}\wedge B _{ij}^\prime=(A_i\vee A_j)\wedge (A_i^\prime\vee A_j^\prime);\;\;\;\{i,j,k\}=\{1,2,3\}.
\end{eqnarray}

We also consider the 
\begin{eqnarray}
{\cal C}_{i}=A_i \vee A_i^\prime.
\end{eqnarray}

If ${\cal C}_1\wedge {\cal C}_2$ is a subset of ${\cal C}_3$, then ${\mathfrak B}_3$ is a subset of ${\mathfrak B}_1\vee {\mathfrak B}_2$.
\begin{eqnarray}\label{23}
{\cal C}_1\wedge {\cal C}_2\prec {\cal C}_3\;\;\rightarrow \;\;{\mathfrak B}_3 \prec {\mathfrak B}_1\vee {\mathfrak B}_2.
\end{eqnarray}
The converse is not true.
\end{proposition}
\begin{proof}
We will work with the negation of these statements.
We assume that ${\cal C}_1\wedge {\cal C}_2\prec {\cal C}_3$. The negation of both sides of this assumption, gives
\begin{eqnarray}\label{61}
{\cal C}_1\wedge {\cal C}_2\prec {\cal C}_3\;\;\rightarrow \;\;
(\neg A_1\wedge \neg A_1^\prime) \vee (\neg A_2\wedge \neg A_2^\prime) \succ \neg A_3 \wedge \neg A_3^\prime.
\end{eqnarray}
We want to prove that
\begin{eqnarray}
&&{\mathfrak B}_3 \prec {\mathfrak B}_1\vee {\mathfrak B}_2\;\;\rightarrow \;\;\nonumber\\
&&(\neg A_1\wedge \neg A_2)\vee (\neg A_1^\prime\wedge \neg A_2^\prime)\succ
[(\neg A_1\wedge \neg A_3)\vee (\neg A_1^\prime\wedge \neg A_3^\prime)]\wedge 
[(\neg A_2\wedge \neg A_3)\vee (\neg A_2^\prime\wedge \neg A_3^\prime)]
\end{eqnarray}
We rewrite this as
\begin{eqnarray}
&&(\neg A_1\wedge \neg A_2)\vee (\neg A_1^\prime\wedge \neg A_2^\prime)\succ\nonumber\\&&
(\neg A_1\wedge \neg A_2\wedge \neg A_3)\vee (\neg A_1^\prime\wedge  \neg A_2^\prime \wedge \neg A_3^\prime)\vee
\{[(\neg A_1^\prime\wedge \neg A_2)\vee  (\neg A_2^\prime\wedge \neg A_1)]\wedge (\neg A_3\wedge \neg A_3 ^\prime)\} 
\end{eqnarray}
Using the assumption of Eq.(\ref{61}), and the fact that $A\prec B$ implies that $A\wedge C\prec B\wedge C$ (the converse of which is not true), we rewrite this as
\begin{eqnarray}
&&(\neg A_1\wedge \neg A_2)\vee (\neg A_1^\prime\wedge \neg A_2^\prime)\succ\nonumber\\&&
[\neg A_1\wedge \neg A_2\wedge (\neg A_3 \vee \neg A_1^\prime\vee \neg A_2^\prime)]\vee
[\neg A_1^\prime\wedge \neg A_2^\prime\wedge (\neg A_3 ^\prime \vee\neg A_1\vee \neg A_2)]
\end{eqnarray}
which is clearly true. We stress that the proof relies on the distributivity property, which holds in Boolean algebra, but does not hold in 
the modular lattices below.
The converse of this proposition is not true, because  $A\wedge C\prec B\wedge C$ does not imply $A\prec B$.
\end{proof}

\begin{figure}
\centering
\begin{tikzpicture}[circuit logic US,
                    tiny circuit symbols,
                    every circuit symbol/.style={fill=white,draw, logic gate input sep=4mm}
]

\node [or gate, inputs = nn] at (0,0) (or1) {OR};
\node [or gate, inputs = nn] at ($(or1.south)+(0,2cm)$) (or2) {OR};
\node [or gate, inputs = nn] at ($(or2.south)+(0,2cm)$) (or3) {OR};
\node [and gate, inputs = nn, anchor=input 1] at ($(or3.south)+(2cm,0)$) (and1) {AND};
\node [or gate, inputs = nn] at ($(and1.south)+(3cm,0cm)$) (or4) {OR};
\draw (or1.input 1) -- ++(left:5mm) node[left] (A) {$A_3$};
\draw (or1.input 2) -- ++(left:5mm) node[left] {$A_3'$};
\draw (or2.input 1) -- ++(left:5mm) node[left] (A) {$A_2$};
\draw (or2.input 2) -- ++(left:5mm) node[left] {$A_2'$};
\draw (or3.input 1) -- ++(left:5mm) node[left] (A) {$A_1$};
\draw (or3.input 2) -- ++(left:5mm) node[left] {$A_1'$};

\draw (or3.output) -- ++(right:3mm) |- (and1.input 1) node[above left] {${\cal C}_1$};
\draw (or2.output) -- ++(right:3mm) |- (and1.input 2)  node[below left] {${\cal C}_2$};

\draw (and1.output) -- ++(right:3mm) |- (or4.input 1)  node[above left] {${\cal C}_1\wedge{\cal C}_2$};
\draw (or1.output) -- ++(right:3mm) node[below right] {${\cal C}_3$} |- (or4.input 2) ;

\draw (or4.output) -- ++(right:5mm);

\end{tikzpicture}
\caption{The left hand side of Eq.(\ref{23}) holds, when the output of this logical circuit is ${\cal C}_3=A_3 \vee A_3^\prime$. } \label{fig2}
\end{figure}
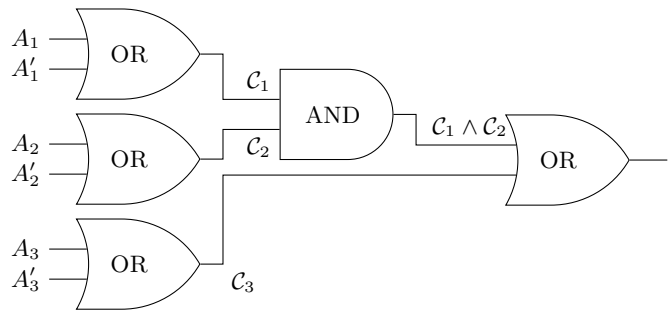 

\begin{figure}
\centering
\begin{tikzpicture}[circuit logic US,
                    tiny circuit symbols,
                    every circuit symbol/.style={fill=white,draw, logic gate input sep=4mm}
]

\node (A3p) at (0,0) {$A_3'$};
\node (A3) at ($(A3p.south) + (0,1.6cm)$) {$A_3$};
\node (A2p) at ($(A3.south) + (0,1.6cm)$) {$A_2'$};
\node (A2) at ($(A2p.south) + (0,1.6cm)$) {$A_2$};
\node (A1p) at ($(A2.south) + (0,1.6cm)$) {$A_1'$};
\node (A1) at ($(A1p.south) + (0,1.6cm)$) {$A_1$};

\node [or gate, inputs = nn] at (3cm,0) (or1) {OR};
\node [or gate, inputs = nn] at ($(or1.south)+(0,2cm)$) (or2) {OR};
\node [or gate, inputs = nn] at ($(or2.south)+(0,2cm)$) (or3) {OR};
\node [or gate, inputs = nn] at ($(or3.south)+(0,2cm)$) (or4) {OR};
\node [or gate, inputs = nn] at ($(or4.south)+(0,2cm)$) (or5) {OR};
\node [or gate, inputs = nn] at ($(or5.south)+(0,2cm)$) (or6) {OR};

\node [and gate, inputs = nn, anchor=input 1] at ($(or2.south)+(2cm,0)$) (and1) {AND};
\node [and gate, inputs = nn, anchor=input 1] at ($(or4.south)+(2cm,0)$) (and2) {AND};
\node [and gate, inputs = nn, anchor=input 1] at ($(or6.south)+(2cm,0)$) (and3) {AND};

\node [or gate, inputs = nn] at ($(and1.south)+(2.5cm,2cm)$) (or7) {OR};

\node [and gate, inputs = nn, anchor=input 1] at ($(or7.south)+(2.5cm,3cm)$) (and4) {AND};
\draw (A1) --++ (right:11mm)|- (or6.input 1) node[above left] {$A_1$};
\draw (A1) --++ (right:11mm)|- (or4.input 1) node[above left] {$A_1$};
\draw (A2) --++ (right:12mm)|- (or6.input 2) node[below left] {$A_2$};
\draw (A2) --++ (right:12mm)|- (or2.input 1) node[above left] {$A_2$};
\draw (A3) --++ (right:13mm)|- (or4.input 2) node[below left] {$A_3$};
\draw (A3) --++ (right:13mm)|- (or2.input 2) node[below left] {$A_3$};
\draw (A1p) --++ (right:14mm)|- (or5.input 1) node[above left] {$A_1'$};
\draw (A1p) --++ (right:14mm)|- (or3.input 1) node[above left] {$A_1'$};
\draw (A2p) --++ (right:15mm)|- (or1.input 1) node[above left] {$A_2'$};
\draw (A2p) --++ (right:15mm)|- (or5.input 2) node[below left] {$A_2'$};
\draw (A3p) --++ (right:10mm)|- (or1.input 2) node[below left] {$A_3'$};
\draw (A3p) --++ (right:10mm)|- (or3.input 2) node[below left] {$A_3'$};

\draw (or1.output) -- ++(right:3mm) |- (and1.input 2) node[below left] {$B_{23}'$};
\draw (or2.output) -- ++(right:3mm) |- (and1.input 1) node[above left] {$B_{23}$};
\draw (or3.output) -- ++(right:3mm) |- (and2.input 2) node[below left] {$B_{13}'$};
\draw (or4.output) -- ++(right:3mm) |- (and2.input 1) node[above left] {$B_{13}$};
\draw (or5.output) -- ++(right:3mm) |- (and3.input 2) node[below left] {$B_{12}'$};
\draw (or6.output) -- ++(right:3mm) |- (and3.input 1) node[above left] {$B_{12}$};

\draw (and1.output) -- ++(right:3mm) |- (or7.input 2) node[below left] {$\mathfrak{B}_{1}$};
\draw (and2.output) -- ++(right:3mm) |- (or7.input 1) node[above left] {$\mathfrak{B}_{2}$};

\draw (or7.output) -- ++(right:1mm) |- (and4.input 2) node[below left] {$\mathfrak{B}_{1}\vee \mathfrak{B}_{2}$};
\draw (and3.output) -- ++(right:3mm) |- (and4.input 1) node[above left] {$\mathfrak{B}_{3}$};
\draw (and4.output) -- ++(right:5mm);

\end{tikzpicture}
\caption{The right hand side of Eq.(\ref{23}) holds, when the output of this logical circuit is ${\mathfrak B}_3=
(A_1\vee A_2)\wedge (A_1^\prime\vee A_2^\prime)$.} \label{fig3}
\end{figure}
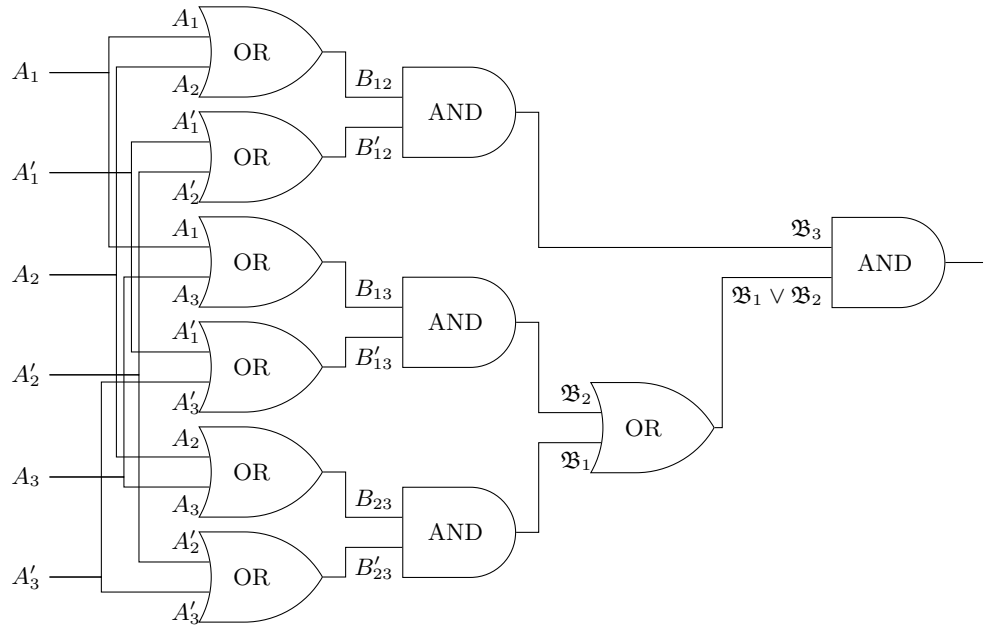 

In terms of gates figs.\ref{fig2}, \ref{fig3} show two different logical circuits related to the left and right hand sides of Eq.(\ref{23}).
The input in the two experiments is the same.
The Desargues property in a classical Physics context, states that if the output in the circuit of fig.\ref{fig2} is ${\cal C}_3$ then the output in the circuit of fig.\ref{fig3} is ${\mathfrak B_3}$.
We note that in general the two outputs are not correlated.
Only if the output of the circuit of fig.\ref{fig2} is ${\cal C}_3$ then the output in the circuit of fig.\ref{fig3} is ${\mathfrak B_3}$.
We use the term {\em selective correlations}, to describe this.
\begin{example}
We consider the case where $S=\{1,2,3\}$ and
\begin{eqnarray}
A_1=\{1\};\;\;\;A_2=\{2\};\;\;\;A_3=\{3\};\;\;\;A_1^\prime=\{1,3\};\;\;\;A_2^\prime=\{2,3\};\;\;\;A_3^\prime=\emptyset.
\end{eqnarray}
In this case,  in the circuit of fig.\ref{fig2} ${\cal C}_3=\{3\}$, and the overall output is ${\cal C}_3$.
Also  in the circuit of fig.\ref{fig3} ${\mathfrak B}_3=\{1,2\}$, and the overall output is ${\mathfrak B}_3$.
It is seen that if the output in the first circuit is ${\cal C}_3$, then the output of the second circuit is ${\mathfrak B}_3$.
But in general the two outputs are not correlated. 
\end{example}

\section{Finite quantum systems}
\subsection{The modular orthocomplemented lattice ${\cal L}(d)$}

We consider a quantum system $\Sigma (d)$ with variables in ${\mathbb Z}(d)$ (the integers modulo $d$), with states in a $d$-dimensional
Hilbert space $H(d)$\cite{vour,vou}.
In the set of subspaces of $H(d)$, we define the conjunction and disjunction\cite{la1,la2,la3,la4,la5}
\begin{eqnarray}
H_1\wedge H_2=H_1\cap H_2;\;\;\;\;\;H_1\vee H_2={\rm span}(H_1 \cup H_2).
\end{eqnarray}
These two operations define the logical `AND' and `OR' in a quantum context. The set of subspaces of $H(d)$
with these operations is a lattice, which we denote as ${\cal L}(d)$.
The corresponding partial order $\prec$ is `subspace'.
The smallest element is ${\cal O}=H(0)$ (the zero-dimensional subspace that contains only the zero vector), and 
the largest element is ${\cal I}=H(d)$. 

It is known that ${\cal L}(d)$ is a modular orthocomplemented lattice.
Modularity is a weak version of distributivity and it states that
\begin{eqnarray}
H_1\prec H_3\;\rightarrow\; H_1\vee (H_2\wedge H_3)=(H_1\vee H_2)\wedge H_3.
\end{eqnarray}
The orthocomplement of $H_1$ (logical NOT operation) is another subspace which we denote as 
$H_1^{\perp}$, with the properties
\begin{eqnarray}\label{ortho}
&&H_1\wedge H_1^{\perp}={\cal O};\;\;\;\;H_1\vee H_1^{\perp}={\cal I};\;\;\;\;(H_1^{\perp})^{\perp}=H_1\nonumber\\
&&(H_1\wedge H_2)^{\perp}=H_1^{\perp}\vee H_2^{\perp};\;\;\;\;(H_1\vee H_2)^{\perp}=H_1^{\perp}\wedge H_2^{\perp}.
\end{eqnarray}
We will use the notation $\Pi(H)$ for the projector to the subspace $H$. Clearly ${\rm Tr}[ \Pi(H)]=\dim (H)$.
Also
\begin{eqnarray}
\Pi(H_1)+\Pi(H_1^\perp)={\bf 1};\;\;\;\Pi(H_1)\Pi(H_1^\perp)=0;\;\;\;\dim (H_1)+\dim (H_1^\perp)=d.
\end{eqnarray}

An important property of modular lattices \cite{la1}, is that
\begin{eqnarray}\label{12}
\dim (H_1\vee H_2)+\dim (H_1\wedge H_2)=\dim (H_1)+\dim (H_2).
\end{eqnarray}

If $H_0$ is a subspace of $H(d)$, we consider the sublattice ${\cal L}(H_0)$ of ${\cal L}(d)$ that contains all subspaces $H_1$ such that ${\cal O}\prec H_1\prec  H_0$.
The orthocomplement of a space $H_1$ in ${\cal L}(H_0)$, is also called relative orthocomplement of $H_1$ in ${\cal L}(d)$.
In this case
\begin{eqnarray}\label{6}
&&H_1\vee H_1^{\perp }=H_0;\;\;\;\Pi(H_1)+\Pi(H_1^\perp)=\Pi(H_0);\;\;\;\dim (H_1)+\dim (H_1^\perp)=\dim (H_0)\nonumber\\
&&\Pi(H_1)\Pi(H_1^\perp)=0;\;\;\;\Pi(H_1)\Pi (H_0)=\Pi(H_1);\;\;\;\Pi(H_1^\perp)\Pi (H_0)=\Pi(H_1^\perp).
\end{eqnarray}
It will be clear from the context, which orthocomplement we consider. 
The relative orthocomplement is a limited logical NOT operation, within ${\cal L}(H_0)$.

The following proposition will be used below to calculate projectors.
\begin{proposition}\label{pro}
Let $v_1,...,v_n$ be $n$ linearly independent vectors in $H(d)$ (where $n\le d$), and $A$ the $d\times n$ matrix $(v_1,...,v_n)$ which has as columns the vectors $v_1,...v_n$.
The projector to the space spanned by these $n$ vectors is
\begin{eqnarray}
\Pi=A(A^\dagger A)^{-1}A^\dagger.
\end{eqnarray}
\end{proposition}
\begin{proof}
We first point out that ${\rm rank} (A)=n$, and therefore ${\rm rank}(A^\dagger A)=n$, which proves that the matrix $A^\dagger A$ is invertible.
It is easily seen that $\Pi^2=\Pi$, and therefore $\Pi$ is a projector, to the space spanned by the vectors $v_1,...,v_n$.
\end{proof}

\subsection{The Desargues property and selective quantum correlations}

We introduce a `dictionary' which shows the analogy between Projective Geometry and Quantum Physics, and which will be used to `translate' 
Desargues theorem in Projective Geometry, into Desargues theorem in Quantum Physics:
\begin{itemize}
\item
We call points the one-dimensional subspaces, lines the two-dimensional subspaces of $H(d)$, and planes the three-dimensional subspaces of $H(d)$. 
We will denote with lower and upper case letters, the points and lines, correspondingly. 
\begin{itemize}
\item
A point $h$ contains a single quantum state, and ${\rm Tr}[ \Pi(h)]=\dim (h)=1$.
\item
A line $H$ contains two quantum states and all their superpositions, and ${\rm Tr}[ \Pi(H)]=\dim (H)=2$.
\end{itemize}
\item
The central concept of projective geometry, is a `duality' between points and lines.
For every statement that involves points, there is a corresponding statement that involves lines, with the 
conjunction (logical AND) and disjunction (logical OR) exchanging roles.
In a quantum context, if we have an expression that contains one-dimensional spaces (points) and two-dimensional spaces (lines) within a three-dimensional subspace $H_0$,
the relative orthocomplement  with respect to $H_0$, exchanges the roles of points and lines, and it also exchanges the roles of 
conjunction and disjunction (Eqs.(\ref{ortho}),(\ref{6})).
\item
If $h_1, h_2$ are distinct one-dimensional subspaces of $H(d)$, 
the disjunction $h_1\vee h_2$ is a line through the points $h_1, h_2$. 
Indeed, in this case Eq.(\ref{12}) gives $\dim (h_1\vee h_2)=\dim (h_1)+\dim (h_2)-\dim (h_1\wedge h_2)=1+1-0=2$.
Physically, $h_1\vee h_2$ contains all the superpositions of the two quantum states corresponding to $h_1$ and $h_2$.
\item
If $H_1, H_2$ are two-dimensional subspaces of $H(d)$ such that $\dim (H_1\vee H_2)=3$,
the conjunction $H_1\wedge H_2$ is a point at the intersection of the lines $H_1, H_2$.
Indeed, in this case Eq.(\ref{12}) gives $\dim (H_1\wedge H_2)=\dim (H_1)+\dim (H_2)-\dim (H_1\vee H_2)=2+2-3=1$.
\item
If $H_1=h_1^\perp$, is the relative orthocomplement of $h_1$ with respect to some subspace $H_0$, 
then $H_1$ and $h_1$ represent the negation of each other (relative logical NOT within the $H_0$).
In terms of projectors $\Pi(H_1)+\Pi (h_1)=\Pi (H_0)$ and $\Pi(H_1)\Pi (h_1)=0$.

\end{itemize}

It is known\cite{A0,A1,A2,A3} that the elements of ${\cal L}(d)$ obey the Desargues property, which is inspired by the well known Desargues theorem in Projective Geometry (e.g. \cite{pro}). It is also known that there exist modular lattices which do not obey the Desargues property.
This is an example of the fact that the lattices related to finite quantum systems, have stronger properties than the
modular orthocomplemented lattices.

We first introduce `triangles':
\begin{itemize}
\item
Let $h_1, h_2, h_3$ be distinct one-dimensional subspaces of $H(d)$ (points).
\begin{itemize}
\item
In the {\em generic case} that ${\dim}(h_1\vee h_2\vee h_3)=3$, we say that the points $h_1, h_2, h_3$ form a triangle. 
\item
If ${\dim}(h_1\vee h_2\vee h_3)=2$, then the three vectors corresponding to these one-dimensional spaces, are linearly dependent.
\end{itemize}
\item
Let $H_1, H_2, H_3$ be distinct two-dimensional subspaces of $H(d)$ (lines). 
\begin{itemize}
\item
In the {\em generic case} that
$ \dim(H_1\wedge H_2\wedge H_3)=0$ (i.e., the three lines $H_1, H_2, H_3$ do not intersect at the same point), we say that the lines $H_1, H_2, H_3$ form a triangle.
\item
If $\dim ( H_1\wedge H_2\wedge H_3)=1$, i.e., if the three lines $H_1, H_2, H_3$ intersect at the same point,
then the three vectors corresponding to their orthocomplements $H_1^\perp, H_2^\perp, H_3^\perp$, are linearly dependent.
\end{itemize}
\end{itemize}
We next introduce some notation.
Let $(h_1, h_2, h_3)$ and $(h_1^\prime, h_2^\prime , h_3^\prime)$ be two triangles which for simplicity we assume to be on the same plane, i.e., 
\begin{eqnarray}\label{35}
h_1\vee h_2\vee h_3=h_1^\prime\vee h_2^\prime \vee h_3^\prime=H_0;\;\;\;\dim (H_0)=3.
\end{eqnarray} 
This assumption is adopted for convinience, and it is not essential.
Also let 
\begin{eqnarray}
H_1=h_1 ^\perp;\;\;\;H_2= h_2 ^\perp;\;\;\;H_3= h_3 ^\perp;\;\;\;H_1^\prime= (h_1^\prime) ^\perp;\;\;\;H_2^\prime= (h_2^\prime) ^\perp;\;\;\;H_3^\prime= (h_3^\prime) ^\perp,
\end{eqnarray}
be their relative orthocomplements, with respect to the subspace $H_0$.
They are two-dimensional spaces (lines).

We consider the lines
\begin{eqnarray}
H_{ij}=h_i\vee h_j;\;\;\;H_{ij}^\prime=h_i^\prime\vee h_j^\prime;\;\;\;i\ne j
\end{eqnarray}
through the vertices in each of these triangles,
and the points at the intersection of a side in the first triangle with a side in the second triangle
\begin{eqnarray}
{\mathfrak h}_k=H_{ij}\wedge H _{ij}^\prime=(h_i\vee h_j)\wedge (h_i^\prime \vee h_j^\prime );\;\;\;\{i,j,k\}=\{1,2,3\}.
\end{eqnarray}
Physically, ${\mathfrak h}_k$ contains a state which is both a superposition of the two quantum states contained in $h_i$ and $h_j$, and also
a superposition of the two quantum states contained in $h_i^\prime$ and $h_j^\prime$.

We also consider the lines that join a vertex in the first triangle with a vertex in the second triangle:
\begin{eqnarray}
{\cal H}_{i}=h_i\vee h_i^\prime.
\end{eqnarray}
Physically, ${\cal H}_{i}$ contains all the superpositions of the two quantum states corresponding to $h_i$ and $h_i^\prime$.

The following proposition gives the Desargues property in a quantum context:
\begin{proposition}\label{pro2}
If the lines ${\cal H}_1, {\cal H}_2, {\cal H}_3$ intersect at the same point, 
then the points ${\mathfrak h}_1, {\mathfrak h}_2, {\mathfrak h}_3$ are on the same line, and vice-versa.
We can express this in different equivalent ways:
\begin{itemize}
\item[(1)]
\begin{eqnarray}\label{A}
\dim ({\cal H}_1\wedge {\cal H}_2\wedge {\cal H}_3)=1\;\;\leftrightarrow \;\;\dim ({\mathfrak h}_1\vee {\mathfrak h}_2\vee {\mathfrak h}_3)=2.
\end{eqnarray}
\item[(2)]
If ${\cal H}_1\wedge {\cal H}_2$ is a subspace of ${\cal H}_3$, then ${\mathfrak h}_3$ is a subspace of ${\mathfrak h}_1\vee {\mathfrak h}_2$:
\begin{eqnarray}\label{5}
{\cal H}_1\wedge {\cal H}_2\prec {\cal H}_3\;\;\leftrightarrow \;\;{\mathfrak h}_3 \prec {\mathfrak h}_1\vee {\mathfrak h}_2.
\end{eqnarray}
This can also be expressed in terms of projectors as follows. If
$\Pi({\cal H}_3)\Pi({\cal H}_1\wedge {\cal H}_2)=\Pi({\cal H}_1\wedge {\cal H}_2)$ then $\Pi({\mathfrak h}_1\vee {\mathfrak h}_2)\Pi({\mathfrak h}_3 )=\Pi({\mathfrak h}_3 )$, and vice-versa.
\end{itemize}
\end{proposition}
\begin{proof}
The proof is given in refs\cite{la3,A1,A2,A3}.
\end{proof}
\begin{remark}
The condition $\Pi({\cal H}_3)\Pi({\cal H}_1\wedge {\cal H}_2)=\Pi({\cal H}_1\wedge {\cal H}_2)$ is stronger than the commutativity condition
$[\Pi({\cal H}_3),\Pi({\cal H}_1\wedge {\cal H}_2)]=0$, in the sense that the former implies the latter, but not vice-versa.
Similarly, the condition $\Pi({\mathfrak h}_1\vee {\mathfrak h}_2)\Pi({\mathfrak h}_3 )=\Pi({\mathfrak h}_3 )$, is stronger than the commutativity condition
$[\Pi({\mathfrak h}_3 ), \Pi({\mathfrak h}_1\vee {\mathfrak h}_2)]=0$.
\end{remark}
In connection with the above proposition, we discuss the following two experiments:
\begin{itemize}
\item
On a state $\ket{s}$ we perform measurement with the projector $\Pi({\cal H}_1\wedge {\cal H}_2)$.
This will give the outcome `yes' with probability $p_1={\bra s}\Pi({\cal H}_1\wedge {\cal H}_2)\ket{s}$, in which case the state ${\ket s}$ will collapse to the state
\begin{eqnarray}
\ket{s_1}=\frac{1}{\sqrt{p_1}}\Pi({\cal H}_1\wedge {\cal H}_2)\ket{s}.
\end{eqnarray}
We then perform a second measurement with the projector $\Pi({\cal H}_3)$, and we get the outcome `yes' with probability $q_1=\bra {s_1}\Pi({\cal H}_3)\ket{s_1}$,
in which case the state $\ket {s_1}$ will collapse to the state
\begin{eqnarray}
\ket{s_2}=\frac{1}{\sqrt{q_1}}\Pi({\cal H}_3)\ket{s_1}=\frac{1}{\sqrt{p_1q_1}}\Pi({\cal H}_3)\Pi({\cal H}_1\wedge {\cal H}_2)\ket{s}.
\end{eqnarray}
We show this schematically:
\begin{eqnarray}\label{EXP1}
\ket{s}\overset {\Pi({\cal H}_1\wedge {\cal H}_2)}\longrightarrow \ket{s_1}\overset{\Pi({\cal H}_3)}\longrightarrow \ket{s_2}.
\end{eqnarray}
\item
On the same  state $\ket{s}$ we perform measurement with the projector $\Pi({\mathfrak h}_3 )$.
This will give the outcome `yes' with probability $p_2={\bra s}\Pi({\mathfrak h}_3 )\ket{s}$, in which case the state ${\ket s}$ will collapse to the state
\begin{eqnarray}
\ket{t_1}=\frac{1}{\sqrt{p_2}}\Pi({\mathfrak h}_3 )\ket{s}.
\end{eqnarray}
We then perform a second measurement with the projector $\Pi({\mathfrak h}_1\vee {\mathfrak h}_2)$,
and we get the outcome `yes' with probability $q_2=\bra {t_1}\Pi({\mathfrak h}_1\vee {\mathfrak h}_2)\ket{t_1}$,
in which case the state $\ket {t_1}$ will collapse to the state
\begin{eqnarray}
\ket{t_2}=\frac{1}{\sqrt{q_2}}\Pi({\mathfrak h}_1\vee {\mathfrak h}_2)\ket{t_1}=\frac{1}{\sqrt{p_2q_2}}\Pi({\mathfrak h}_1\vee {\mathfrak h}_2)
\Pi({\mathfrak h}_3 )\ket{s}.
\end{eqnarray}
We show this schematically:
\begin{eqnarray}\label{EXP2}
\ket{s}\overset {\Pi({\mathfrak h}_3 )}\longrightarrow \ket{t_1}\overset{\Pi({\mathfrak h}_1\vee {\mathfrak h}_2)}\longrightarrow \ket{t_2}.
\end{eqnarray}
\end{itemize}
Desargues property says that there is the following selective correlation between these two experiments.
If $\Pi({\cal H}_3)\Pi({\cal H}_1\wedge {\cal H}_2)=\Pi({\cal H}_1\wedge {\cal H}_2)$
and therefore
$\ket{s_1}=\ket{s_2}$ in the first experiment, then 
$\Pi({\mathfrak h}_1\vee {\mathfrak h}_2)\Pi({\mathfrak h}_3 )=\Pi({\mathfrak h}_3 )$ and therefore
$\ket{t_1}=\ket{t_2}$ in the second experiment (and vice-versa).
We note that in this case $\ket{s_1}=\ket{s_2}$ is an eigenstate of $\Pi({\cal H}_3)$,
 and $\ket{t_1}=\ket{t_2}$ is an eigenstate of $\Pi({\mathfrak h}_1\vee {\mathfrak h}_2)$.
In the general case, the two outputs are not correlated.
The two experiments in Eqs(\ref{EXP1}),(\ref{EXP2}), can be performed in different locations.

Various quantities that quantify correlations (e.g., quantum discord \cite{OZ}, etc) might be used for further analysis of   
selective correlations.
\subsection{Example}
In $H(5)$ we consider the one-dimensional subspaces $(h_1, h_2, h_3, h_1^\prime, h_2^\prime , h_3^\prime)$ that contain the vectors
\begin{eqnarray}
h_1=\begin{pmatrix}
0\\
1\\
1+i\\
2\\
0\\
\end{pmatrix};\;\;\;
h_2=\begin{pmatrix}
0\\
1\\
0\\
2\\
0\\
\end{pmatrix};\;\;\;
h_3=\begin{pmatrix}
0\\
1\\
1+i\\
0\\
0\\
\end{pmatrix};\;\;\;
h_1^\prime=\begin{pmatrix}
0\\
1\\
3\\
2\\
0\\
\end{pmatrix};\;\;\;
h_2^\prime=\begin{pmatrix}
0\\
1-i\\
1+i\\
2\\
0\\
\end{pmatrix};\;\;\;
h_3^\prime=\begin{pmatrix}
0\\
1-i\\
-1-i\\
4-2i\\
0\\
\end{pmatrix}
\end{eqnarray}
In order to simplify the notation, we represent a space with a `generic vector' that it contains.

We easily show that $\dim (h_1\vee h_2\vee h_3)=3$, by finding numerically that the rank of the $5\times 3$ matrix $(v_1,v_2,v_3)$ is $3$.
In a similar way we find that $\dim (h_1^\prime\vee h_2^\prime \vee h_3^\prime)=3$.
Also the fact that the first and last component of all vectors is zero, shows that Eq.(\ref{35}) is satisfied.
Therefore we have two triangles $(h_1, h_2, h_3)$ and $(h_1^\prime, h_2^\prime , h_3^\prime)$, on the same `plane', which is the three-dimensional space $H_0$, with generic vector
\begin{eqnarray}
H_0=\begin{pmatrix}
0\\
\alpha\\
\beta\\
\gamma\\
0\\
\end{pmatrix}.
\end{eqnarray}

The lines ( two-dimensional spaces) ${\cal H}_1=h_1\vee h_1^\prime$,  ${\cal H}_2=h_2\vee h_2^\prime$, ${\cal H}_3=h_3\vee h_3^\prime$, contain the superpositions
\begin{eqnarray}\label{10}
&&{\cal H}_1=h_1\vee h_1^\prime=\begin{pmatrix}
0\\
a_1+b_1\\
a_1+a_1i+3b_1\\
2a_1+2b_1\\
0\\
\end{pmatrix};\;\;\;
{\cal H}_2=h_2\vee h_2^\prime=\begin{pmatrix}
0\\
a_2+b_2-b_2i\\
b_2+b_2i\\
2a_2+2b_2\\
0\\
\end{pmatrix}\nonumber\\&&
{\cal H}_3=h_3\vee h_3^\prime=\begin{pmatrix}
0\\
a_3+b_3-b_3i\\
a_3+a_3i-b_3-b_3i\\
4b_3-2b_3i\\
0\\
\end{pmatrix}
\end{eqnarray}
The three lines ${\cal H}_1, {\cal H}_2, {\cal H}_3$ go through the point
\begin{eqnarray}\label{15}
w=\begin{pmatrix}
0\\
2-i\\
0\\
4-2i\\
0\\
\end{pmatrix}
\end{eqnarray}
This is easily seen if we use the values
\begin{eqnarray}
a_1=3;\;\;\;b_1=-1-i;\;\;\;a_2=2-i;\;\;\;b_2=0;\;\;\;a_3=1;\;\;\;b_3=1,
\end{eqnarray}
in Eq.(\ref{10}).

We next calculate the lines (two-dimensional spaces) $H_{12}=h_1\vee h_2$, $H_{13}=h_1\vee h_3$, $H_{23}=h_2\vee h_3$:
They contain the superpositions:
\begin{eqnarray}
&&H_{12}=h_1\vee h_2=\begin{pmatrix}
0\\
A_1+A_2\\
A_1+A_1i\\
2A_1+2A_2\\
0\\
\end{pmatrix};\;\;\;
H_{13}=h_1\vee h_3=\begin{pmatrix}
0\\
A_3+A_4\\
A_3+A_3i+A_4+A_4i\\
2A_3\\
0\\
\end{pmatrix}\nonumber\\
&&H_{23}=h_2\vee h_3=\begin{pmatrix}
0\\
A_5+A_6\\
A_6+A_6i\\
2A_5\\
0\\
\end{pmatrix}
\end{eqnarray}
Similarly we calculate the lines
(two-dimensional spaces) $H_{12}^\prime=h_1^\prime\vee h_2^\prime$, $H_{13}^\prime=h_1^\prime\vee h_3^\prime$, $H_{23}^\prime=h_2^\prime\vee h_3^\prime$:
They contain the superpositions:
\begin{eqnarray}
&&H_{12}^\prime=h_1^\prime\vee h_2^\prime=\begin{pmatrix}
0\\
A_1^\prime+A_2^\prime-A_2^\prime i\\
3A_1^\prime+A_2^\prime +A_2^\prime i\\
2A_1^\prime+2A_2^\prime\\
0\\
\end{pmatrix};\;\;\;
H_{13}^\prime=h_1^\prime\vee h_3^\prime=\begin{pmatrix}
0\\
A_3^\prime+A_4^\prime-A_4^\prime i\\
3A_3^\prime-A_4^\prime-A_4^\prime i\\
2A_3^\prime+4A_4^\prime-2A_4^\prime i\\
0\\
\end{pmatrix}\nonumber\\
&&H_{23}^\prime=h_2^\prime\vee h_3^\prime=\begin{pmatrix}
0\\
A_5^\prime-A_5^\prime i+A_6^\prime-A_6^\prime i\\
A_5^\prime+A_5^\prime i-A_6^\prime-A_6^\prime i\\
2A_5^\prime+4A_6^\prime-2A_6^\prime i\\
0\\
\end{pmatrix}
\end{eqnarray}
We next find the intersections:
\begin{eqnarray}\label{19}
{\mathfrak h}_3=H_{12}\wedge H _{12}^\prime=
\begin{pmatrix}
0\\
1\\
3\\
2\\
0\\
\end{pmatrix};\;\;\;
{\mathfrak h}_2=H_{13}\wedge H _{13}^\prime=
\begin{pmatrix}
0\\
1\\
1+i\\
3\\
0\\
\end{pmatrix};\;\;\;
{\mathfrak h}_1=H_{23}\wedge H _{23}^\prime=
\begin{pmatrix}
0\\
1-i\\
-1-i\\
4-2i\\
0\\
\end{pmatrix}
\end{eqnarray}
We can easily check that these three points are on the same line because
\begin{eqnarray}
{\rm rank} ({\mathfrak h}_1, {\mathfrak h}_2, {\mathfrak h}_3)=2.
\end{eqnarray}
It is seen that Eq.(\ref{A}) holds.

An alternative approach is to show that the relation that involves projectors in proposition \ref{pro2}, holds.
In order to calculate $\Pi({\cal H}_3)$, we chose two vectors in the space ${\cal H}_3$ by putting $a_3=1$, $b_3=0$ and also $a_3=0$, $b_3=1$
in the general vector in Eq.(\ref{10}). This gives the matrix $A$ for the projector formula in proposition \ref{pro}. We get
\begin{eqnarray}
A=\begin{pmatrix}
0& 0\\
1 &1-i\\
1+i& -1-i\\
0 &4-2i\\
0 &0\\
\end{pmatrix}\;\;\rightarrow\;\;\Pi({\cal H}_3)=
\begin{pmatrix}
        0          &        0          &        0          &        0      &            0        \\  
        0           &  0.4286 & 0.2857 - 0.2857i &  0.2857   &    0        \\  
        0         &    0.2857 + 0.2857i  & 0.7143& -0.1429 - 0.1429i   &     0       \\   
        0          &   0.2857      &      -0.1429 + 0.1429i &  0.8571      &            0        \\  
        0          &        0           &       0       &           0            &      0        
\end{pmatrix}
\end{eqnarray}
The one-dimensional space ${\cal H}_1\wedge {\cal H}_2$ consists of the vector in Eq.(\ref{15}), and therefore
\begin{eqnarray}
A=\begin{pmatrix}
0\\
2-i\\
0\\
4-2i\\
0\\
\end{pmatrix}\;\;\rightarrow\;\;\Pi({\cal H}_1\wedge {\cal H}_2)=\begin{pmatrix}
        0        &          0         &         0           &       0             &     0    \\      
        0        &   0.2    &    0    &         0.4  &     0          \\
        0             &     0       &           0       &           0      &            0      \\    
        0        &     0.4   &    0        &     0.8     &   0          \\
        0          &        0          &        0       &           0         &         0     
\end{pmatrix}.
\end{eqnarray}
The one-dimensional space ${\mathfrak h}_3$ is given in Eq.(\ref{19}), and therefore
\begin{eqnarray}
A=\begin{pmatrix}
0\\
1\\
3\\
2\\
0\\
\end{pmatrix}\;\;\rightarrow\;\;\Pi({\mathfrak h}_3)=\begin{pmatrix}
  0     &    0    &     0    &     0      &   0\\
         0 &   0.0714 &   0.2143   & 0.1429  &       0\\
         0  & 0.2143  &  0.6429 &   0.4286     &    0\\
         0   & 0.1429 &   0.4286 &   0.2857   &      0\\
         0   &      0    &     0     &    0  &       0
\end{pmatrix}.
\end{eqnarray}
The $A$ matrix for the two-dimensional space ${\mathfrak h}_1\vee {\mathfrak h}_2$ can be found from Eq.(\ref{19}), where vectors in ${\mathfrak h}_1, {\mathfrak h}_2$ are given:
\begin{eqnarray}
A=\begin{pmatrix}
0&0\\
1&1-i\\
1+i&-1-i\\
3&4-2i\\
0&0\\
\end{pmatrix}\;\;\rightarrow\;\;\Pi({\mathfrak h}_1\vee {\mathfrak h}_2)=\begin{pmatrix}
 0        &          0           &       0        &          0        &          0  \\        
        0     &        0.1017   &          0.1186 + 0.0339i  & 0.2712 - 0.0508i   &     0      \\    
        0     &        0.1186 - 0.0339i &  0.9831      &      -0.0339 + 0.0169i   &     0    \\      
        0     &        0.2712 + 0.0508i & -0.0339 - 0.0169i  & 0.9153     &             0   \\       
        0        &          0     &             0         &         0       &           0
\end{pmatrix}.
\end{eqnarray}
It is easily verified that the relations between the projectors given in proposition \ref{pro2}, hold.
Consequently with these projectors, measurements on an arbitrary vector $\ket{s}$ will give $\ket{s_1}=\ket{s_2}$ and $\ket{t_1}=\ket{t_2}$.
For example, if
\begin{eqnarray}
\ket{s}=\begin{pmatrix}
   0.2294\\
    0.4588\\
    0.2294\\
    0.6882\\
    0.4588\\
\end{pmatrix}
\end{eqnarray}
we get $p_1=0.673$, $p_2=0.454$, $q_1=q_2=1$ and 
\begin{eqnarray}
\ket{s_1}=\ket{s_2}=\begin{pmatrix}
         0\\
    0.4472\\
         0\\
    0.8944\\
         0\\
\end{pmatrix};\;\;\;
\ket{t_1}=\ket{t_2}=\begin{pmatrix}
         0\\
    0.2673\\
    0.8018\\
    0.5345\\
         0\\
\end{pmatrix}.
\end{eqnarray}

In the general case that the three lines ${\cal H}_1, {\cal H}_2, {\cal H}_3$ do not intersect at the same point,
the outputs in the two experiments are uncorrelated. So we have selective correlations between the outputs of the two experiments, only in the case that the lines
${\cal H}_1, {\cal H}_2, {\cal H}_3$ go through the same point, as in the example above (where they go through the point in Eq.(\ref{15})).

\section{Discussion}

The Desargues theorem is fundamental in Projective Geometry\cite{pro}.
An analogue to this is discussed here, in the context of both Classical and Quantum Physics.
It shows the existence of selective correlations, which could be useful in classical and quantum technologies. 

In a classical context, the  Desargues property is expressed in proposition \ref{pro1}.
From an applied point of view, it implies the existence of selective correlations in the two logical circuits 
in figs.\ref{fig1},\ref{fig2}, each of which has six inputs $A_1, A_2, A_3, A_1^\prime , A_2^\prime, A_3^\prime$.
If the output of the first one is ${\cal C}_3=A_3 \vee A_3^\prime$,
then the output of the second one is ${\mathfrak B}_3=(A_1\vee A_2)\wedge (A_1^\prime\vee A_2^\prime)$.
But in general the two outputs are not correlated. 

In a quantum context, the  Desargues property is expressed in proposition \ref{pro2}.
From an applied point of view, it implies the existence of selective correlations in the two experiments shown 
schematically in Eqs(\ref{EXP1}),(\ref{EXP2}).
If in one experiment the second measurement does not change the output of the first measurement, then the same is true in the other experiment.
This holds only for particular set of projectors, defined in proposition \ref{pro2}.

The Desargues theorem has been used in a quantum context different from ours, in ref.\cite{Rau}.

The quantum logic of finite quantum systems is described by modular orthocomplemented lattices.
But this is a rather weak statement, because these systems have stronger properties, like
the Desargues property discussed here.
An even stronger statement is that the quantum logic of finite quantum systems is described by
the lattices of commuting equivalence relations (sometimes called linear lattices)\cite{A0,A1,A2,A3,A4},
which are not discussed here.
Linear lattices are modular orthocomplemented lattices, with many
extra properties that involve not only the Desargues theorem, but also its many generalizations.
The interpretation of all these properties in a quantum context in terms of quantum correlations, is an open problem.
The present paper is an important first step in this direction, because it discusses the Desargues theorem which is the 
fundamental starting point. 
  
This whole area provides the theoretical foundation for both classical (Boolean) computation, and quantum computation.

\end{document}